\documentclass[11pt]{article}

\usepackage{mathpazo}
\usepackage{amsfonts}
\usepackage{amsmath}
\usepackage{graphicx}
\usepackage{latexsym}

\usepackage[margin=1.05in]{geometry}

\usepackage[T1]{fontenc}
\usepackage{times}
\usepackage{color,graphicx}
\usepackage{array}
\usepackage{enumerate}
\usepackage{amsmath}
\usepackage{amssymb}
\usepackage{amsthm}
\usepackage{pgfplots}
\usepackage{pgf}
\usepackage{tikz}
\usepackage{filecontents}
\usepgfplotslibrary{patchplots} 
\usetikzlibrary{pgfplots.patchplots} 
\pgfplotsset{width=9cm,compat=1.5.1}

\usepackage{amsthm}

\newtheorem{lemma}{Lemma}
\newtheorem{theorem}{Theorem}
\newtheorem{corollary}{Corollary}

\newtheorem{example}{Example}

\usepackage{xcolor}
\usepackage{makeidx}
\usepackage[colorlinks=true,linkcolor=blue,anchorcolor=blue,citecolor=red,urlcolor=magenta]{hyperref}

\newcommand{\tr}{\operatorname{Tr}}

\newcommand{\bra}[1]{\langle #1 |}
\newcommand{\ket}[1]{| #1 \rangle}

\newcommand{\ketbra}[2]{| #1 \rangle\langle #2 |}

\newcommand\ip[2]{\ensuremath{\langle#1,#2\rangle}}

\newcommand{\defeq}{\stackrel{\smash{\textnormal{\tiny def}}}{=}}

\begin{document}

\title{The Inverse Eigenvalue Problem for Entanglement Witnesses}

\author{
	Nathaniel Johnston\footnote{Department of Mathematics \& Computer Science, Mount Allison University, Sackville, NB, Canada E4L 1E4}\textsuperscript{$\ \ $}\footnote{Department of Mathematics \& Statistics, University of Guelph, Guelph, ON, Canada N1G 2W1}
	\quad and\quad
	Everett Patterson\footnotemark[1]
}

\date{August 20, 2017}

\maketitle
\begin{abstract}
	We consider the inverse eigenvalue problem for entanglement witnesses, which asks for a characterization of their possible spectra (or equivalently, of the possible spectra resulting from positive linear maps of matrices). We completely solve this problem in the two-qubit case and we derive a large family of new necessary conditions on the spectra in arbitrary dimensions. We also establish a natural duality relationship with the set of absolutely separable states, and we completely characterize witnesses (i.e., separating hyperplanes) of that set when one of the local dimensions is $2$.
\end{abstract}

\section{Introduction}

In linear algebra and matrix theory, an inverse eigenvalue problem asks for a characterization of the possible spectra (i.e., the ordered tuples of eigenvalues) of a given set of matrices. Perhaps the most well-known inverse eigenvalue problem asks for the possible spectra of entrywise non-negative matrices \cite{ELN04}. This is called the non-negative inverse eigenvalue problem (NIEP), and it has been completely solved for matrices of size $4 \times 4$ and smaller, but remains unsolved for larger matrices.

Several variants of the NIEP have also been investigated, where instead a characterization is asked for of the possible spectra of \emph{symmetric} non-negative matrices \cite{Sou83}, stochastic matrices \cite{DD46,Joh81}, or doubly stochastic matrices \cite{PM65,HP04}. Similarly, the inverse eigenvalue problem has been considered for Toeplitz matrices \cite{Lan94} and tridiagonal matrices \cite{PSES07}, among many others (see \cite{Chu98} and the references therein).

In this paper, we consider the inverse eigenvalue problem for entanglement witnesses, which are matrices of interest in quantum information theory that will be defined in the next section. Equivalently, we investigate what spectra can result from applying a positive matrix-valued map to just part of a positive semidefinite matrix. Such maps are of interest in operator theory (and again, we introduce the mathematical details in the next section).

The paper is organized as follows. In Section~\ref{sec:math_prelims}, we introduce the various mathematical tools that we will use throughout the paper. In Section~\ref{sec:two_qubit}, we completely solve the inverse eigenvalue problem for two-qubit entanglement witnesses (i.e., the lowest-dimension non-trivial case). In Section~\ref{sec:qubit_qudit}, we extend our investigation to qubit-qudit entanglement witnesses and obtain a large family of new necessary conditions on the spectra. In the process, we completely characterize the witnesses of the set of absolutely separable states in these dimensions. In Section~\ref{sec:high_dim_ews}, we extend our results to obtain necessary conditions on the spectra of decomposable entanglement witnesses in arbitrary dimensions. Finally, we provide closing remarks and open questions in Section~\ref{sec:conclude}.

\section{Mathematical Preliminaries}\label{sec:math_prelims}

\subsection{Separability, Entanglement, and the Partial Transpose}

From a mathematical perspective, quantum information theory (and more specifically, quantum entanglement theory) is largely concerned with properties of (Hermitian) positive semidefinite matrices and the tensor product. A \emph{pure quantum state} $\ket{v} \in \mathbb{C}^n$ is a unit (column) vector and a \emph{mixed quantum state} $\rho \in M_n(\mathbb{C})$ is a (Hermitian) positive semidefinite matrix with $\tr(\rho) = 1$ (we use $\tr(\cdot)$ to denote the trace and $M_n(\mathbb{R})$ or $M_n(\mathbb{C})$ to denote the set of $n \times n$ real or complex matrices, respectively). Whenever we use ``ket'' notation like $\ket{v}$ or $\ket{w}$, or lowercase Greek letters like $\rho$ or $\sigma$, we are implicitly assuming that they represent pure or mixed quantum states, respectively.

A mixed state $\rho \in M_m(\mathbb{C}) \otimes M_n(\mathbb{C})$ is called \emph{separable} if there exist pure states $\{\ket{v_j}\} \subseteq \mathbb{C}^m$ and $\{\ket{w_j}\} \subseteq \mathbb{C}^n$ such that
\[
	\rho = \sum_{j} p_j \ketbra{v_j}{v_j} \otimes \ketbra{w_j}{w_j},
\]
where ``$\otimes$'' is the tensor (Kronecker) product, $\bra{v}$ is the dual (row) vector of $\ket{v}$, so $\ketbra{v}{v}$ is the rank-$1$ projection onto $\ket{v}$, and $p_1,p_2,\ldots$ form a probability distribution (i.e., they are non-negative and add up to $1$). Equivalently, $\rho$ is separable if and only if it can be written in the form
\[
	\rho = \sum_j X_j \otimes Y_j,
\]
where each $X_j \in M_m(\mathbb{C})$ and $Y_j \in M_n(\mathbb{C})$ is a (Hermitian) positive semidefinite matrix. If $\rho$ is not separable then it is called \emph{entangled}.

Determining whether or not a mixed state is separable is a hard problem \cite{Gur03,Gha10}, so in practice numerous one-sided tests are used to demonstrate separability or entanglement (see \cite{GT09,HHH09} and the references therein for a more thorough introduction to this problem). The most well-known such test says that if we define the \emph{partial transpose} of a matrix $A = \sum_j B_j \otimes C_j \in M_m(\mathbb{C}) \otimes M_n(\mathbb{C})$ via
\[
	A^\Gamma := \sum_j B_j \otimes C_j^T,
\]
then $\rho$ being separable implies that $\rho^\Gamma$ is positive semidefinite (so we write $\rho^\Gamma \succeq O$) \cite{Per96}. This test follows simply from the fact that if $\rho$ is separable then $\rho = \sum_j X_j \otimes Y_j$ with each $X_j,Y_j \succeq O$, so
\[
	\rho^\Gamma = \sum_j X_j \otimes Y_j^T,
\]
which is still positive semidefinite since each $Y_j^T$ is positive semidefinite, and tensoring and adding positive semidefinite matrices preserves positive semidefiniteness. If a mixed state $\rho$ is such that $\rho^\Gamma \succeq O$ then we say that it has \emph{positive partial transpose (PPT)}, and the previous discussion shows that the set of separable states is a subset of the set of PPT states.

\subsection{Entanglement Witnesses and Positive Maps}

A straightforward generalization of the partial transpose test for entanglement is based on positive linear matrix-valued maps. A linear map $\Phi : M_m(\mathbb{C}) \rightarrow M_n(\mathbb{C})$ is called \emph{positive} if $X \succeq O$ implies $\Phi(X) \succeq O$, and the transpose is an example of one such map. Based on positive maps, we can define \emph{block-positive matrices}, which are matrices of the form $W := (id_m \otimes \Phi)(X)$ for some $O \preceq X \in M_m(\mathbb{C}) \otimes M_m(\mathbb{C})$ and some positive linear map $\Phi : M_m(\mathbb{C}) \rightarrow M_n(\mathbb{C})$. Equivalently, $W$ is block-positive if and only if $(\bra{a}\otimes\bra{b})W(\ket{a}\otimes\ket{b}) \geq 0$ for all $\ket{a} \in \mathbb{C}^m$ and all $\ket{b} \in \mathbb{C}^n$.

If $W$ is block-positive but not positive semidefinite, it is called an \emph{entanglement witness}, since it is then the case that $\tr(W\sigma) \geq 0$ for all separable $\sigma \in M_m(\mathbb{C}) \otimes M_n(\mathbb{C})$, but there exists some (necessarily entangled) mixed state $\rho \in M_m(\mathbb{C}) \otimes M_n(\mathbb{C})$ such that $\tr(W\rho) < 0$. That is, $W$ verifies or ``witnesses'' the fact that $\rho$ is entangled (geometrically, $W$ acts as a separating hyperplane that separates $\rho$ from the convex set of separable states). A matrix is called \emph{decomposable} if it can be written in the form $W = X^\Gamma + Y$, where $X,Y \succeq O$. It is straightforward to verify that every decomposable matrix is block-positive. However, the converse of this statement (i.e., every block-positive matrix is decomposable) is true if and only if $(m,n) = (2,2)$, $(2,3)$, or $(3,2)$ \cite{Sto63,Wor76}. We denote the set of block positive and decomposable matrices in $M_m(\mathbb{C}) \otimes M_n(\mathbb{C})$ by $\textup{BP}_{m,n}$ and $\textup{DBP}_{m,n}$, respectively.

Our primary interest in this work is characterizing the possible spectra of entanglement witnesses. However, it is a bit more natural to work with the set of block-positive matrices, since it is closed and convex (neither of which is true of the set of entanglement witnesses). However, this distinction does not matter much, since any spectral inequality that we obtain for the block-positive matrices can be turned into a spectral inequality for entanglement witnesses by just adding the condition ``at least one of the eigenvalues is strictly negative''. With this in mind, we are now in a position to define the two main sets that we will be investigating throughout the rest of this paper:
\begin{align*}
	\sigma(\textup{BP}_{m,n}) & \defeq \big\{(\mu_1,\mu_2,\ldots,\mu_{mn}) : \text{$\exists \ W \in \textup{BP}_{m,n}$ with eigenvalues $\mu_1,\mu_2,\ldots,\mu_{mn}$} \big\} \\
	\sigma(\textup{DBP}_{m,n}) & \defeq \big\{(\mu_1,\mu_2,\ldots,\mu_{mn}) : \text{$\exists \ W \in \textup{DBP}_{m,n}$ with eigenvalues $\mu_1,\mu_2, \ldots, \mu_{mn}$} \big\}.
\end{align*}

In words, $\sigma(\textup{BP}_{m,n})$ and $\sigma(\textup{DBP}_{m,n})$ are the sets of possible spectra of block-positive matrices and decomposable block-positive matrices, respectively. We emphasize that we do \emph{not} require the vectors in these sets to be sorted---if a particular vector is in one of those sets, then so is every vector obtained by permuting its entries. However, it will sometimes be convenient to refer to the ordered eigenvalues of a block-positive matrix, so we sometimes use the notation $\vec{\mu}^\downarrow$ to refer to the vector with the same entries as $\vec{\mu}$, but sorted in non-increasing order (i.e., $\mu_1^\downarrow \geq \mu_2^\downarrow \geq \cdots \geq \mu_{mn}^\downarrow$).

\subsection{Known Results on Spectra of Entanglement Witnesses}\label{sec:known_results}

We now summarize all known results concerning the spectrum of a (decomposable) block-positive matrix $W \in M_m(\mathbb{C}) \otimes M_n(\mathbb{C})$ that we are aware of:

\begin{itemize}
	\item $W$ has no more than $(m-1)(n-1)$ negative eigenvalues \cite{Par04,Sar08}, and this number of negative eigenvalues is attainable even if $W$ is decomposable \cite{Joh13}. In particular, if $m = n = 2$ then every entanglement witness has exactly $1$ negative eigenvalue.
	
	\item $\lambda_{\textup{min}}(W) / \lambda_{\textup{max}}(W) \geq 1-\min\{m,n\}$ \cite{JK10,Joh12}.
	
	\item If $W$ has $q$ negative eigenvalues then \cite{Joh12}:
	\begin{align*}
		\frac{\lambda_{\textup{min}}(W)}{\lambda_{\textup{max}}(W)} & \geq 1 - \frac{mn\sqrt{mn - 1}}{q\sqrt{mn - 1} + \sqrt{mnq - q^2}} \quad \text{and} \\
		\frac{\lambda_{\textup{min}}(W)}{\lambda_{\textup{max}}(W)} & \geq 1 - \left\lceil \frac{1}{2}\Big(m+n - \sqrt{(m-n)^2 + 4q - 4}\Big) \right\rceil.
	\end{align*}
	
	\item $\tr(W)^2 \geq \tr(W^2)$ \cite{SWZ08} (this is a spectral condition since $\tr(W) = \sum_j \lambda_j(W)$ and $\tr(W^2) = \sum_j \lambda_j(W)^2$).
	
	\item If $W$ is decomposable then $\lambda_{\textup{min}}(W) \geq -\tr(W)/2$ \cite{Ran13}. 
\end{itemize}

There is one family of block-positive matrices whose eigenvalues are particularly simple to analyze, and those are the matrices of the form $(\ketbra{v}{v})^\Gamma$, where $\ket{v} \in \mathbb{C}^m \otimes \mathbb{C}^n$. The following lemma is well-known (see \cite{Hil07} for example), but we prove it for completeness. Note that the lemma relies on the \emph{Schmidt coefficients} of $\ket{v}$, which are the singular values of $\ket{v}$ when it is thought of as an $m \times n$ matrix (this is a standard tool in quantum information theory, so the reader is directed to a textbook like \cite{NC00} for further details).

\begin{lemma}\label{lem:eigenvalues_rank1_pt}
	Suppose $\ket{v} \in \mathbb{C}^m \otimes \mathbb{C}^n$ and for simplicity assume that $m \leq n$. If $\ket{v}$ has Schmidt coefficients $\alpha_1 \geq \alpha_2 \geq \cdots \geq \alpha_m \geq 0$ then the matrix $(\ketbra{v}{v})^\Gamma$ has eigenvalues
	\begin{align*}
	\alpha_j^2 & \quad \text{for} \quad 1 \leq j \leq m, \\
	\pm\alpha_i\alpha_j & \quad \text{for} \quad 1 \leq i \neq j \leq m, \quad \text{and} \\
	0 & \quad \text{with extra multiplicity $m(n-m)$}.
	\end{align*}
\end{lemma}

\begin{proof}
	By virtue of the Schimdt decomposition, we can write
	\[
		\ket{v} = \sum_{j=1}^m \alpha_j \ket{b_j} \otimes \ket{c_j},
	\]
	where $\{\ket{b_j}\} \subseteq \mathbb{C}^m$ and $\{\ket{c_j}\} \subseteq \mathbb{C}^n$ are orthonormal sets of vectors. Straightforward computation shows that
	\[
		(\ketbra{v}{v})^\Gamma = \sum_{i,j=1}^m \alpha_i \alpha_j  \ketbra{b_i}{b_j} \otimes \overline{ \ketbra{c_j}{c_i} }.
	\]
	To find the eigenvalues, we first define the following (eigen)vectors:
	\begin{align*}
		\ket{x_j} & := \ket{b_j} \otimes \overline{ \ket{c_j} } \quad \text{for} \quad 1 \leq j \leq m \qquad \text{and} \\
		\ket{y_{i,j}^\pm} & := \frac{1}{\sqrt{2}}\big(\ket{b_i} \otimes \overline{ \ket{c_j} } \pm \ket{b_j} \otimes \overline{ \ket{c_i} }\big) \quad \text{for} \quad 1 \leq i \neq j \leq m.
	\end{align*}
	Then direct computation shows that
	\begin{align*}
		(\ketbra{v}{v})^\Gamma \ket{x_j} = \alpha_j^2 \ket{x_j} \qquad \text{and} \qquad (\ketbra{v}{v})^\Gamma \ket{y_{i,j}^\pm} = \pm \alpha_i \alpha_j \ket{y_{i,j}^\pm},
	\end{align*}
	which establishes the claim about the potentially non-zero eigenvalues of $(\ketbra{v}{v})^\Gamma$. To see that there are $m(n-m)$ extra $0$ eigenvalues (for a total of $m + m(m-1) + m(n-m) = mn$ eigenvalues), first extend $\{\ket{c_j}\}$ to an orthonormal basis $\{\ket{c_1},\ldots,\ket{c_m},\ket{d_1},\ldots,\ket{d_{n-m}}\}$ of $\mathbb{C}^n$. Then define the (eigen)vectors $\ket{z_{i,j}} := \ket{b_i} \otimes \ket{d_j}$ for $1 \leq i \leq m$ and $1 \leq j \leq n-m$. It is then straightforward to check that $(\ketbra{v}{v})^\Gamma\ket{z_{i,j}} = 0\ket{z_{i,j}}$ for all $i,j$, which completes the proof.
\end{proof}

\subsection{Cones and Semidefinite Programming}

One nice feature of the sets $\sigma(\textup{BP}_{m,n})$ and $\sigma(\textup{DBP}_{m,n})$ is that they are closed (this is not hard to prove, but the argument is done explicitly for $\sigma(\textup{BP}_{m,n})$ in \cite{GCM14}) and they are \emph{cones}: if $c \in \mathbb{R}$ is a non-negative scalar and $\vec{\mu} \in \sigma(\textup{BP}_{m,n})$, then $c\vec{\mu} \in \sigma(\textup{BP}_{m,n})$ too (and similarly for $\sigma(\textup{DBP}_{m,n})$). Given a particular cone $\mathcal{C} \subseteq \mathbb{R}^n$, its \emph{dual cone} $\mathcal{C}^\circ$ is defined by
\begin{align*}
	\mathcal{C}^\circ \defeq \big\{ \vec{x} \in \mathbb{R}^n : \langle \vec{x},\vec{y}\rangle \geq 0 \ \text{for all} \ \vec{y} \in \mathcal{C} \big\}.
\end{align*}

It is well-known that the double-dual of any closed cone $\mathcal{C} \subseteq \mathbb{R}^n$ is its convex hull: $\mathcal{C}^{\circ\circ} = \mathrm{Conv}(\mathcal{C})$ \cite{BV04}. For this reason, in this work it will often be much easier to work with $\mathrm{Conv}(\sigma(\textup{BP}_{m,n}))$ and $\mathrm{Conv}(\sigma(\textup{DBP}_{m,n}))$ instead of $\sigma(\textup{BP}_{m,n})$ and $\sigma(\textup{DBP}_{m,n})$ directly. Before proceeding, it is worth demonstrating that $\sigma(\textup{BP}_{m,n})$ and $\sigma(\textup{DBP}_{m,n})$ are indeed not convex:

\begin{example}\label{exam:eigs_not_convex}
	Consider the positive semidefinite matrix
	\begin{align*}
	X = \begin{bmatrix}
	1 & 0 & 0 & 2 \\
	0 & 0 & 0 & 0 \\
	0 & 0 & 0 & 0 \\
	2 & 0 & 0 & 4
	\end{bmatrix} \in M_2(\mathbb{C}) \otimes M_2(\mathbb{C}).
	\end{align*}
	Then
	\begin{align*}
	W := X^\Gamma = \begin{bmatrix}
	1 & 0 & 0 & 0 \\
	0 & 0 & 2 & 0 \\
	0 & 2 & 0 & 0 \\
	0 & 0 & 0 & 4
	\end{bmatrix}
	\end{align*}
	is a block-positive matrix with eigenvalues $(4,2,1,-2) \in \sigma(\textup{BP}_{2,2})$. Thus
	\begin{align*}
	\frac{1}{2}(4,2,1,-2) + \frac{1}{2}(4,2,-2,1) = (4,2,-1/2,-1/2) \in \mathrm{Conv}\big(\sigma(\textup{BP}_{2,2})\big).
	\end{align*}
	But as we mentioned in Section~\ref{sec:known_results}, it is well-known that a block-positive matrix in $\textup{BP}_{2,2}$ can have at most one negative eigenvalue, which shows that $\sigma(\textup{BP}_{2,2})$ is not convex. Similar examples can also be constructed in higher dimensions. 
\end{example}

One useful tool for probing closed convex cones is \emph{semidefinite programming}, which is a method of optimizing a linear function over constraints involving positive semidefinite matrices. We do not give a full introduction to semidefinite programming here (see \cite{BV04} for such an introduction), but rather we note that it contains linear programming as a special case. For our purposes, given a linear map $L : \mathbb{R}^{m} \rightarrow M_{n}(\mathbb{R})$, a vector $\vec{v} \in \mathbb{R}^m$, and a symmetric matrix $A \in M_n(\mathbb{R})$, the associated semidefinite program is the following pair of optimization problems:
\begin{center}
	\begin{minipage}{2in}
		\centerline{\underline{Primal problem}}\vspace{-7mm}
		\begin{align*}
		\text{minimize:}\quad & \ip{\vec{v}}{\vec{x}}\\
		\text{subject to:}\quad & L(\vec{x}) \succeq A,\\
		& \vec{0} \leq \vec{x} \in \mathbb{R}^m 
		\end{align*}
	\end{minipage}
	\hspace*{1.5cm}
	\begin{minipage}{2.4in}
		\centerline{\underline{Dual problem}}\vspace{-7mm}
		\begin{align}\begin{split}\label{sdp:form}
		\text{maximize:}\quad & \tr(AY)\\
		\text{subject to:}\quad & L^{*}(Y) \leq \vec{v},\\
		& O \preceq Y \in M_n(\mathbb{R}),
		\end{split}\end{align}
	\end{minipage}
\end{center}
where $L^* : M_{n}(\mathbb{R}) \rightarrow \mathbb{R}^{m}$ is the \emph{dual map} of $L$, defined by the fact that $\tr(L(\vec{x})Y^*) = \ip{\vec{x}}{L^*(Y)}$ for all $\vec{x} \in \mathbb{R}^{m}$ and $Y \in M_n(\mathbb{R})$.

\emph{Weak duality} for semidefinite programs says that $\ip{\vec{v}}{\vec{x}} \geq \tr(AY)$ whenever $\vec{x}$ and $Y$ satisfy the constraints of the semidefinite program. \emph{Strong duality} says that, under slightly stronger assumptions, we can find a particular $\vec{x}$ and $Y$ so that this inequality becomes an equality. The following theorem provides one possible set of assumptions that lead to strong duality (see \cite[Lecture~7]{Wat11}, for example):

\begin{theorem}[Slater conditions for strong duality]\label{thm:slater}
	If there exists $\vec{x} > \vec{0}$ such that $L(\vec{x}) \succ A$ and $Y \succ O$ such that $L^*(Y) < \vec{v}$, then strong duality holds for the semidefinite program~\eqref{sdp:form}. That is, there exists an feasible point $\vec{x_0}$ of the primal problem and a feasible point $Y_0$ of the dual problem such that $\ip{\vec{v}}{\vec{x_0}} = \tr(AY_0)$, and this quantity is the optimal value of both optimization problems.
\end{theorem}

In words, Slater's theorem says that strong duality holds for any semidefinite program in which both the primal problem and the dual problem are \emph{strictly feasible}.

\subsection{Absolute Separability and Absolute PPT}

The final ingredient that we need to be able to prove our results is absolute separability. A state $\rho \in M_m(\mathbb{C}) \otimes M_n(\mathbb{C})$ is called \emph{absolutely separable} \cite{KZ00} if $U\rho U^*$ is separable for all unitary matrices $U \in M_m(\mathbb{C}) \otimes M_n(\mathbb{C})$. Similarly, a state is said to be \emph{absolutely PPT} \cite{Hil07} if $U\rho U^*$ has positive partial transpose (PPT) for all unitary matrices $U \in M_m(\mathbb{C}) \otimes M_n(\mathbb{C})$. Since both absolute separability and absolute PPT only depend on the spectrum of the state $\rho$, we define these sets via those spectra instead of via the states:
\begin{align*}
	\textup{ASEP}_{m,n} & \defeq \big\{ (\lambda_1,\ldots,\lambda_{mn}) : \text{$\exists$ abs. sep. $\rho \in M_m(\mathbb{C}) \otimes M_n(\mathbb{C})$ with eigenvalues $\lambda_1,\ldots,\lambda_{mn}$} \big\}, \\
	\textup{APPT}_{m,n} & \defeq \big\{ (\lambda_1,\ldots,\lambda_{mn}) : \text{$\exists$ abs. PPT $\rho \in M_m(\mathbb{C}) \otimes M_n(\mathbb{C})$ with eigenvalues $\lambda_1,\ldots,\lambda_{mn}$} \big\}.
\end{align*}

In the case when one of the local dimensions is $2$, the sets $\textup{ASEP}_{2,n}$ and $\textup{APPT}_{2,n}$ have a simple characterization \cite{Hil07,Joh13b,VAD01}:
\begin{align*}
	\textup{ASEP}_{2,n} = \textup{APPT}_{2,n} & = \big\{ (\lambda_1,\ldots,\lambda_{2n}) : \lambda_1 \leq \lambda_{2n-1} + 2\sqrt{\lambda_{2n-2}\lambda_{2n}} \big\} \\
	& = \left\{ (\lambda_1,\ldots,\lambda_{2n}) : \begin{bmatrix}
		2\lambda_{2n} & \lambda_{2n-1} - \lambda_1 \\ \lambda_{2n-1} - \lambda_1 & 2\lambda_{2n-2}
	\end{bmatrix} \succeq O \right\}.
\end{align*}
When $m,n \geq 3$, the set $\textup{APPT}_{m,n}$ has been completely characterized (again, see \cite{Hil07}), but the characterization is rather complicated, so we leave the details until we need them in Section~\ref{sec:high_dim_ews}. Not much is known about the set $\textup{ASEP}_{m,n}$ when $m,n \geq 3$ other than the obvious fact that $\textup{ASEP}_{m,n} \subseteq \textup{APPT}_{m,n}$. However, it is not even known whether or not this inclusion is strict \cite{AJR15}.

The reason for our interest in absolute separability and absolute PPT in this work is the following result, which establishes a duality result between these problems and the inverse eigenvalue problem for block positive matrices:

\begin{theorem}\label{thm:conv_hull_duality_appt1}
	The following duality relationships hold for all $m,n \geq 1$:
	\begin{align*}
		\sigma(\textup{BP}_{m,n})^\circ = \textup{ASEP}_{m,n} \qquad \text{and} \qquad \sigma(\textup{DBP}_{m,n})^\circ = \textup{APPT}_{m,n}.
	\end{align*}
\end{theorem}

\begin{proof}
	This result for $\textup{APPT}_{m,n}$ was essentially shown (though not explicitly stated in this way) in \cite{Hil07}. With that in mind, we explicitly prove the characterization of $\textup{ASEP}_{m,n}$, and just note that the analogous result for $\textup{APPT}_{m,n}$ can be proved in a very similar manner.
	
	We start by noting that $\vec{\lambda} \in \textup{ASEP}_{m,n}$ if and only if
	\begin{align}\label{eq:asep_dual_cone_test}
		\tr(W(U\mathrm{diag}(\vec{\lambda})U^*)) \geq 0 \quad \text{for all} \quad W \in \textup{BP}_{m,n} \quad \text{and unitary} \quad U \in M_m(\mathbb{C}) \otimes M_n(\mathbb{C}).
	\end{align}
	Well, it is a well-known result (see \cite[Problem~III.6.14]{Bha97}, for example) that $\tr(W(U\mathrm{diag}(\vec{\lambda})U^*)) \geq 0$ for all unitary $U$ if and only if
	\begin{align*}
		\sum_{j=1}^{mn} \lambda_j \mu_{p(j)} \geq 0 \quad \text{for all permutations} \quad p : [mn] \rightarrow [mn],
	\end{align*}
	where $\mu_1, \mu_2, \ldots, \mu_{mn}$ are the eigenvalues of $W$. (In fact, it suffices to check that $\sum_{j=1}^{mn} \lambda_j \mu_{mn+1-j} \geq 0$, but this additional simplification is not relevant for our purposes.)
	
	Since $\sigma(\textup{BP}_{m,n})$ is invariant under permutations of its vector's entries, it follows that condition~\eqref{eq:asep_dual_cone_test} is equivalent to
	\begin{align*}
		\vec{\lambda} \cdot \vec{\mu} = \sum_{j=1}^{mn} \lambda_j \mu_{j} \geq 0 \quad \text{for all} \quad \vec{\mu} \in \sigma(\textup{BP}_{m,n}).
	\end{align*}
	In other words, we have shown that $\vec{\lambda} \in \textup{ASEP}_{m,n}$ if and only if $\vec{\lambda} \in \sigma(\textup{BP}_{m,n})^\circ$, so $\sigma(\textup{BP}_{m,n})^\circ = \textup{ASEP}_{m,n}$, as desired.
\end{proof}

By taking the dual cone of the sets in Theorem~\ref{thm:conv_hull_duality_appt1} and using the fact that the double-dual of a closed cone is its convex hull, we immediately obtain the following corollary:

\begin{corollary}\label{cor:conv_hull_duality_appt}
	The following duality relationships hold for all $m,n \geq 1$:
	\begin{align*}
		\textup{ASEP}_{m,n}^\circ = \mathrm{Conv}\big(\sigma(\textup{BP}_{m,n})\big) \qquad \text{and} \qquad \textup{APPT}_{m,n}^\circ = \mathrm{Conv}\big(\sigma(\textup{DBP}_{m,n})\big).
	\end{align*}
\end{corollary}

It is worth noting that these results show that the sets $\mathrm{Conv}\big(\sigma(\textup{BP}_{m,n})\big)$ and $\mathrm{Conv}\big(\sigma(\textup{DBP}_{m,n})\big)$ are the sets of spectra of what might be called ``absolute separability witnesses'' and ``absolute PPT witnesses'', respectively. These types of witnesses were considered in \cite{GCM14}.

\section{Two--Qubit Entanglement Witnesses}\label{sec:two_qubit}

With the preliminaries out of the way, we now consider the simplest non-trivial entanglement witnesses that exist, which are those that live in $M_2(\mathbb{C}) \otimes M_2(\mathbb{C})$ (a two-dimensional piece of quantum information is called a ``qubit'', so these entanglement witnesses are sometimes called ``two-qubit'' entanglement witnesses). Recall that every block-positive matrix $W \in M_2(\mathbb{C}) \otimes M_n(\mathbb{C})$ (when $n = 2$ or $n = 3$) is decomposable and thus can be written in the form $W = X^\Gamma + Y$, where $X,Y \in M_2(\mathbb{C}) \otimes M_n(\mathbb{C})$ are both positive semidefinite. Before proceeding, we first need the following slight strengthening of this fact in the $n=2$ case:

\begin{lemma}\label{lem:2x2_EW_decomp}
	A matrix $W \in M_2(\mathbb{C}) \otimes M_2(\mathbb{C})$ is block positive if and only if there exist positive semidefinite $X,Y \in M_2(\mathbb{C}) \otimes M_2(\mathbb{C})$ such that $W = X^\Gamma + Y$. Furthermore, $X$ can be chosen to have rank $1$.
\end{lemma}

\begin{proof}
	As we already noted, everything except for the ``furthermore'' remark is already known, so we just need to show that we can choose $X$ to have rank~$1$. To this end, suppose without loss of generality that $W$ is scaled so that $\tr(X) = 1$ (and thus $X$ is a mixed state). Then we use the fact from \cite[Section~IV]{LMO06} that we can write $X$ in the form
	\begin{align*}
	X = \frac{1}{4}(U \otimes V)\left(I + \sum_{k=1}^3 d_k\sigma_k \otimes \sigma_k \right)(U \otimes V)^*,
	\end{align*}
	where $U,V \in M_2(\mathbb{C})$ are unitary, $d_1,d_2,d_3$ are real numbers, and $\sigma_1,\sigma_2,\sigma_3$ are the Pauli matrices
	\begin{align*}
	\sigma_1 := \begin{bmatrix}
	0 & 1 \\ 1 & 0
	\end{bmatrix}, \quad \sigma_2 := \begin{bmatrix}
	0 & -i \\ i & 0
	\end{bmatrix}, \quad \text{and} \quad \sigma_3 := \begin{bmatrix}
	1 & 0 \\ 0 & -1
	\end{bmatrix}.
	\end{align*}
	Since conjugation by $U \otimes V$ has no effect on the properties we are investigating (positive semidefiniteness, being an entanglement witness, having positive partial transpose, and so on), from now on we assume without loss of generality that $X$ has the form
	\begin{align*}
	X = \frac{1}{4}\left(I + \sum_{k=1}^3 d_k\sigma_k \otimes \sigma_k\right).
	\end{align*}
	
	As noted in \cite{LMO06}, $X$ being a quantum state is equivalent to the statement that $(d_1,d_2,d_3)$ is in the convex hull of the four points $(1,1,-1)$, $(1,-1,1)$, $(-1,1,1)$, and $(-1,-1,-1)$. It is straightforward to check that the four matrices corresponding to these points have rank~$1$. For example, if $(d_1,d_2,d_3) = (1,-1,1)$ then direct computation shows that
	\begin{align*}
		\frac{1}{4}\big(I + \sigma_1 \otimes \sigma_1 - \sigma_2 \otimes \sigma_2 + \sigma_3 \otimes \sigma_3\big) = \ketbra{\psi_+}{\psi_+},
	\end{align*}
	where $\ket{\psi_+}$ is the pure state $\ket{\psi_+} := (\ket{0}\otimes\ket{0}+\ket{1}\otimes\ket{1})/\sqrt{2}$.
	
	Well, this convex hull is a tetrahedron in $\mathbb{R}^3$, and similarly the set of $(d_1,d_2,d_3)$ corresponding to (not necessarily positive semidefinite) matrices with positive partial transpose is another tetrahedron. Their intersection is an octahedron that corresponds to the set of PPT states. These tetrahedra and octahedron are depicted in Figure~\ref{fig:tetrahedra}.
	
	\begin{figure}[htb]
		\centering
		\includegraphics[width=6.5cm]{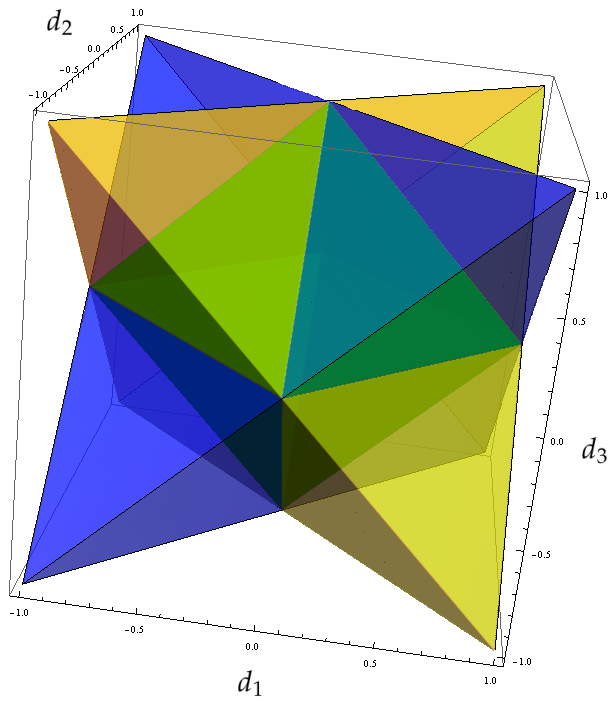}
		\caption{The region of tuples $(d_1,d_2,d_3) \in \mathbb{R}^3$ corresponding to the set of quantum states (blue tetrahedron), partial transposes of quantum states (yellow tetrahedron), and PPT states (octahedron where the two tetrahedra intersect). Every state in the blue tetrahedron can be written as a convex combination of the rank-$1$ state at the nearest vertex of the tetrahedron and some PPT state in the octahedron.}\label{fig:tetrahedra}
	\end{figure}
	
	The result now follows almost immediately from this geometric picture. Since $X$ is in the (blue) tetrahedron of mixed states, it can be written in the form $t\ketbra{v}{v} + (1-t)\sigma$, where $\ketbra{v}{v}$ is the rank-$1$ state at the closest of the four corners of the tetrahedron, and $\sigma$ is a PPT state ($\sigma$ can be chosen to be the closest point in the central octahedron). Thus \[ W = X^\Gamma + Y = (t\ketbra{v}{v} + (1-t)\sigma)^\Gamma + Y = (t\ketbra{v}{v})^\Gamma + ((1-t)\sigma^\Gamma + Y).\]Since $\sigma$ is PPT, $\sigma^\Gamma$ is positive semidefinite, so this is a decomposition of $W$ has the desired form.
\end{proof}

It is worth noting that Lemma~\ref{lem:2x2_EW_decomp} is \emph{not} true in $M_2(\mathbb{C}) \otimes M_3(\mathbb{C})$. As mentioned ealier, we can indeed always write a block positive matrix $W \in M_2(\mathbb{C}) \otimes M_3(\mathbb{C})$ in the form $W = X^\Gamma + Y$, where $X,Y \in M_2(\mathbb{C}) \otimes M_3(\mathbb{C})$ are positive semidefinite. However, the following example shows that we cannot in general choose $X$ to have rank~$1$.

\begin{example}\label{ex:23_2_neg}
	Consider the matrix
	\begin{align*}
		W = \left[
		\begin{array}{ccc|ccc}
		1 & 0 & 0 & 0 & 1 & 0 \\
		0 & 1 & 0 & 0 & 0 & 1 \\
		0 & 0 & 0 & 0 & 0 & 0 \\\hline
		0 & 0 & 0 & 0 & 0 & 0 \\
		1 & 0 & 0 & 0 & 1 & 0 \\
		0 & 1 & 0 & 0 & 0 & 1
		\end{array}
		\right]^\Gamma \in M_2(\mathbb{C}) \otimes M_3(\mathbb{C}).
	\end{align*}
	Since $W$ is the partial transpose of a positive semidefinite matrix, it is block positive. However, direct calculation shows that its eigenvalues are $1$ and $(1 \pm \sqrt{5})/2$, each with multiplicity $2$. Since two of these eigenvalues are negative, it cannot possibly be written in the form $W = X^\Gamma + Y$ with $X,Y$ positive semidefinite and $X$ having rank~$1$, since Lemma~\ref{lem:eigenvalues_rank1_pt} implies that any such $W$ has at most one negative eigenvalue.
\end{example}

We are now able to state the main result of this section, which provides a complete characterization of the eigenvalues of two-qubit block-positive matrices/entanglement witnesses.

\begin{theorem}\label{thm:2x2_eigs}
	There exists a block-positive matrix in $M_2(\mathbb{C}) \otimes M_2(\mathbb{C})$ with eigenvalues $\mu_1, \mu_2, \mu_3, \mu_{4} \in \mathbb{R}$ if and only if the following three inequalities hold:
	\begin{enumerate}[(a)]
		\item $\mu_{3}^\downarrow \geq 0$,
		\item $\mu_{4}^\downarrow \geq -\mu_2^\downarrow$, and
		\item $\mu_{4}^\downarrow \geq -\sqrt{\mu_1^\downarrow\mu_3^\downarrow}$.
	\end{enumerate}
\end{theorem}

\begin{proof}
	We start by proving that the three listed eigenvalue inequalities are necessary.
	
	As we mentioned in Section~\ref{sec:known_results}, inequality~(a) is already known. To see that inequalities~(b) and~(c) are necessary, we start by noting that they both hold for block-positive matrices of the form $W = (\ketbra{v}{v})^\Gamma$, since Lemma~\ref{lem:eigenvalues_rank1_pt} says that the if $\alpha_1,\alpha_2$ are the Schmidt coefficients of $\ket{v}$ then the eigenvalues of $W$ are $\mu_1^\downarrow = \alpha_1^2$, $\mu_2^\downarrow = \alpha_1\alpha_2$, $\mu_3^\downarrow = \alpha_2^2$, and $\mu_4^\downarrow = -\alpha_1\alpha_2$, and it is straightforward to check that these quantities always satisfy inequalities~(b) and~(c).
	
	Well, to see that inequalities~(b) and~(c) are necessary for \emph{all} block-positive matrices, we recall from Lemma~\ref{lem:2x2_EW_decomp} that we can write every block-positive matrix $W \in M_2(\mathbb{C}) \otimes M_2(\mathbb{C})$ in the form
	\[
		W = t(\ketbra{v}{v})^\Gamma + Y,
	\]
	where $t \geq 0$ and $Y$ is positive semidefinite. Well, we just showed that the inequalities $\mu_2^\downarrow + \mu_{4}^\downarrow \geq 0$ and $\mu_{4}^\downarrow + \sqrt{\mu_1^\downarrow\mu_3^\downarrow} \geq 0$ are satisfied by the eigenvalues of $(\ketbra{v}{v})^\Gamma$. Since adding a positive semidefinite matrix $Y$ can only increase the eigenvalues, the eigenvalues of $W$ must also satisfy these same inequalities (i.e., (b) and (c)). This completes the proof of necessity.
	
	To see that the inequalities are sufficient, we demonstrate a procedure for explicitly constructing a block-positive matrix with any set of eigenvalues satisfying the inequalities (a), (b), and (c). If $\mu_4^\downarrow \geq 0$ then all eigenvalues are non-negative so we can find a positive semidefinite $W$ that works, and if $\mu_3^\downarrow = 0$ then inequality~(c) implies $\mu_4^\downarrow = 0$, so we can again find a positive semidefinite $W$. Thus we assume from now on that $\mu_3^\downarrow > 0 > \mu_4^\downarrow$. Define $\alpha_2 := \sqrt{\mu_3^\downarrow}$ and $\alpha_1 := -\mu_4^\downarrow / \alpha_2$, as well as
	\begin{align*}
		\gamma_1 := \mu_1^\downarrow - (\mu_4^\downarrow)^2/\mu_3^\downarrow \quad \text{and} \quad \gamma_2 = \mu_2^\downarrow + \mu_4^\downarrow,
	\end{align*}
	which are all non-negative by inequalities~(b) and~(c). Then it is straightforward to check that the matrix
	\begin{align*}
		W = \begin{bmatrix}\alpha_1^2 & 0 & 0 & \alpha_1\alpha_2 \\ 0 & 0 & 0 & 0 \\ 0 & 0 & 0 & 0 \\ \alpha_1\alpha_2 & 0 & 0 & \alpha_2^2\end{bmatrix}^\Gamma + \gamma_1\begin{bmatrix}1 & 0 & 0 & 0 \\ 0 & 0 & 0 & 0 \\ 0 & 0 & 0 & 0 \\ 0 & 0 & 0 & 0\end{bmatrix} + \frac{\gamma_2}{2}\begin{bmatrix}0 & 0 & 0 & 0 \\ 0 & 1 & 1 & 0 \\ 0 & 1 & 1 & 0 \\ 0 & 0 & 0 & 0\end{bmatrix}
	\end{align*}
	is block-positive and has eigenvalues $\mu_1^\downarrow \geq \mu_2^\downarrow \geq \mu_3^\downarrow \geq \mu_4^\downarrow$ (the corresponding eigenvectors are $\ket{0}\otimes\ket{0}$, $\ket{0}\otimes\ket{1}+\ket{1}\otimes\ket{0}$, $\ket{1}\otimes\ket{1}$, and $\ket{0}\otimes\ket{1} - \ket{1}\otimes\ket{0}$, respectively).
\end{proof}

In order to visualize the set of possible spectra of two-qubit block-positive matrices $W$, note that (since $\tr(W) \geq 0$ and $\tr(W) = 0$ if and only if $W = O$), we can always scale $W$ to have $\tr(W) = 1$. In other words, we can choose $\mu_4 = 1 - \mu_1 - \mu_2 -\mu_3$ and then plot the tuples $(\mu_1,\mu_2,\mu_3) \in \mathbb{R}^3$, where $\mu_1,\mu_2,\mu_3 \geq 0$ are the three non-negative eigenvalues of $W$. This region is displayed in Figure~\ref{fig:two_qubit_ew}---the fact that it is not convex should not be surprising, given Example~\ref{exam:eigs_not_convex}.

\begin{figure}[htb]
\centering
\includegraphics[width=7cm]{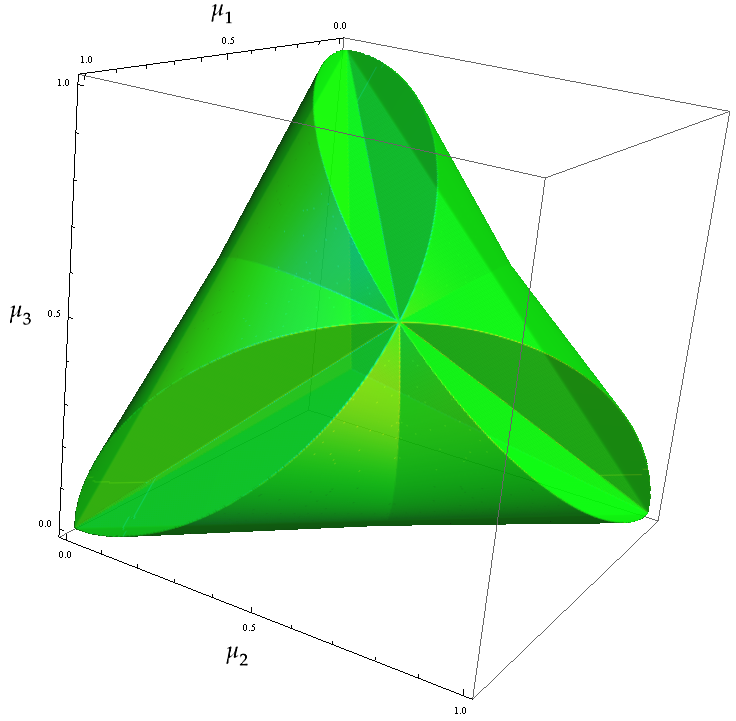}
\caption{The region of tuples $(\mu_1,\mu_2,\mu_3) \in \mathbb{R}^3$ with the property that $(\mu_1,\mu_2,\mu_3,\mu_4)$ is the spectrum of a block-positive matrix, where $\mu_1,\mu_2,\mu_3 \geq 0$ and $\mu_4 = 1-\mu_1-\mu_2-\mu_3$ (i.e., the block-positive matrix is normalized to have trace $1$).}\label{fig:two_qubit_ew}
\end{figure}

\section{Qubit--Qudit Entanglement Witnesses}\label{sec:qubit_qudit}

We now consider entanglement witnesses in the case where just one of the local dimensions is $2$. In this case, we are only able to get a new family of necessary conditions that the spectra must satisfy---not necessary and sufficient conditions like in the previous section. The main result of this section provides a complete characterization of $\mathrm{Conv}\big(\sigma(\textup{BP}_{2,n})\big)$ and $\mathrm{Conv}\big(\sigma(\textup{DBP}_{2,n})\big)$, which in turn imply necessary conditions on $\sigma(\textup{BP}_{2,n})$:

\begin{theorem}\label{thm:mu2n_characterize}
	Suppose $\vec{\mu} \in \mathbb{R}^{2n}$. Define
	\[
		s_k := \sum_{j=k}^{2n} \mu_j^\downarrow \quad \text{for} \quad k = 1, 2, 3 \qquad \text{and} \qquad s_{-} := \frac{1}{2}\left(\sum_{j=1}^{2n} \mu_j - \sum_{j=1}^{2n} |\mu_j|\right) = \sum_{j \ : \ \mu_j < 0} \mu_j.
	\]
	Then the following are equivalent:
	\begin{enumerate}
		\item[a)] $\vec{\mu} \in \mathrm{Conv}\big(\sigma(\textup{BP}_{2,n})\big)$.
		
		\item[b)] $\vec{\mu} \in \mathrm{Conv}\big(\sigma(\textup{DBP}_{2,n})\big)$.
		
		\item[c)] There exists a positive semidefinite matrix $X \in M_2(\mathbb{R})$ such that
		\begin{align*}
		x_{1,1}+x_{2,2} & \leq s_1, & & & x_{2,2} & \leq s_2, \\
		x_{1,2} + x_{2,2} & \leq s_3, & \text{and} & & x_{1,2} & \leq s_{-}.
		\end{align*}
		
		\item[d)] If we define $q_1 := s_1^2 - 4s_{-}^2$ and $q_2 := (s_1+2s_3)^2 - 8s_3^2$ then the following inequalities all hold:
		\begin{align*}
			q_1,q_2 & \geq 0 \\
			\sqrt{q_1} & \geq s_1 - 2s_2 \\
			\sqrt{q_2} & \geq s_1 - 4s_2 + 2s_3 \\
			2\sqrt{q_1} + \sqrt{q_2} & \geq s_1 - 2s_3.
		\end{align*}
	\end{enumerate}
\end{theorem}

\begin{proof}
	The equivalence of (a) and (b) follows immediately from Corollary~\ref{cor:conv_hull_duality_appt} and the fact that $\textup{ASEP}_{2,n} = \textup{APPT}_{2,n}$ \cite{Joh13b}.
	
	To see that (b) and (c) are equivalent, we use Corollary~\ref{cor:conv_hull_duality_appt} to note that $\vec{\mu} \in \mathrm{Conv}\big(\sigma(\textup{DBP}_{2,n})\big)$ if and only if $\vec{\mu} \in \textup{APPT}_{2,n}^\circ$, if and only if the optimal value of the following (primal) semidefinite program is non-negative:
	\begin{align}\begin{split}\label{sdp:2n_primal}
	\text{minimize:} & \ \sum_{j=1}^{2n} \mu_j^\downarrow \lambda_{2n+1-j} \\
	\text{subject to:} & \ \begin{bmatrix}
	2\lambda_{2n} & \lambda_{2n-1} - \lambda_1 \\
	\lambda_{2n-1} - \lambda_1 & 2\lambda_{2n-2}
	\end{bmatrix} \succeq O \\
	& \ \sum_{j=1}^{2n} \lambda_j = 1 \\
	& \ \lambda_1 \geq \lambda_2 \geq \cdots \geq \lambda_{2n} \geq 0.
	\end{split}\end{align}
	
	(Note that the constraint $\sum_{j=1}^{2n} \lambda_j = 1$ does not actually affect whether or not the optimal value is non-negative, but it is easier to demonstrate that strong duality holds with it present). Well, it will be convenient to instead work with the dual of the above SDP, which (after some routine calculation and simplification) has the following form:
	\begin{align}\begin{split}\label{sdp:2n_dual}
	\text{maximize:} & \ d \\
	\text{subject to:} & \ -2b + z_{2n-1} + d \leq \mu_{2n}^\downarrow \\
	& \ - z_k + z_{k-1} + d \leq \mu_k^\downarrow \quad \text{for} \quad k = 4,5,\ldots,2n-1 \\
	& \ 2c-z_3 + z_2 + d \leq \mu_3^\downarrow \\
	& \ 2b-z_2 + z_1 + d \leq \mu_2^\downarrow \\
	& \ 2a - z_1 + d \leq \mu_1^\downarrow \\
	& d,z_1,\ldots,z_{2n-1} \geq 0 \\
	& \ \begin{bmatrix}a & b \\ b & c\end{bmatrix} \succeq O.
	\end{split}\end{align}
	
	To see that strong duality holds for this pair of SDPs, note that in the primal problem we can choose a large scalar $c$ and a small scalar $\varepsilon > 0$, then we can construct a strictly feasible point of the primal problem~\eqref{sdp:2n_primal} by letting $\lambda_j = (c - j\varepsilon)/(2nc)$ for all $j$. Similarly, to construct a strictly feasible point of the dual problem~\eqref{sdp:2n_dual}, we choose $b = 0$, $a=c=z_i = \varepsilon$ for all $i$, and $d$ very large and negative. Thus both the primal and dual problems are strictly feasible, so Slater's theorem (Theorem~\ref{thm:slater}) tells us that strong duality holds, and in particular these SDPs have the same optimal value.
	
	Thus the primal SDP~\eqref{sdp:2n_primal} has non-negative optimal value (i.e., $\vec{\mu} \in \mathrm{Conv}\big(\sigma(\textup{DBP}_{2,n})\big)$) if and only if there is a feasible point of the dual SDP~\eqref{sdp:2n_dual} with $d = 0$. To simplify this dual existence problem into the form described by condition~(c) of the theorem, we note that we can eliminate the variables $z_1, \ldots, z_{2n-1}$ in a straightforward way since all of the constraints involving them are linear. For example, to eliminate $z_{2n-1}$ we note that the only constraints that involve $z_{2n-1}$ are
	\begin{align*}
		\max\big\{0, z_{2n-2} - \mu_{2n-1}^\downarrow\big\} \leq z_{2n-1} \leq \mu_{2n}^\downarrow + 2b.
	\end{align*}
	We can thus ``squeeze out'' $z_{2n-1}$ by noting that it exists if and only if $\max\big\{0, z_{2n-2} - \mu_{2n-1}^\downarrow\big\} \leq \mu_{2n}^\downarrow + 2b$ (i.e., if and only if $0 \leq \mu_{2n}^\downarrow + 2b$ and $z_{2n-2} - \mu_{2n-1}^\downarrow \leq \mu_{2n}^\downarrow + 2b$). After simplifying a bit, we see that we are now asking whether or not there exist $a,b,c,z_1,\ldots,z_{2n-2}$ such that the following constraints are satisfied:
	\begin{align}\begin{split}\label{sdp:2n_exist_constraints}
		-2b & \leq \mu_{2n}^\downarrow \\
		-2b + z_{2n-2} & \leq \mu_{2n-1}^\downarrow + \mu_{2n}^\downarrow \\
		-z_k + z_{k-1} & \leq \mu_k^\downarrow \quad \text{for} \quad k = 4,5,\ldots,2n-2 \\
		2c-z_3 + z_2 & \leq \mu_3^\downarrow \\
		2b-z_2 + z_1 & \leq \mu_2^\downarrow \\
		2a - z_1 & \leq \mu_1^\downarrow \\
		z_1,\ldots,z_{2n-2} & \geq 0 \\
		\begin{bmatrix}a & b \\ b & c\end{bmatrix} & \succeq O.
	\end{split}\end{align}
	
	By repeating this same argument for $k = 2n-2, 2n-3, \ldots, 5, 4$, we similarly see that the pair of constraints
	\begin{align*}
		-2b + z_k \leq \sum_{j=k+1}^{2n}\mu_{j}^\downarrow \quad \text{and} \quad -z_k + z_{k-1} \leq \mu_k^\downarrow
	\end{align*}
	are equivalent to the pair of constraints (no longer depending on $z_k$)
	\begin{align*}
		-2b + z_{k-1} \leq \sum_{j=k}^{2n}\mu_{j}^\downarrow \quad \text{and} \quad -2b \leq \sum_{j=k+1}^{2n}\mu_{j}^\downarrow.
	\end{align*}
	Thus we see that the system of inequalities~\eqref{sdp:2n_exist_constraints} is equivalent to the following system that only asks for the existence of $a,b,c,z_1,z_2,z_3$ such that
	\begin{align*}
		-2b & \leq \sum_{j=k}^{2n}\mu_{j}^\downarrow \quad \text{for} \quad k = 5,\ldots,2n \\
		-2b + z_3 & \leq \sum_{j=4}^{2n}\mu_{j}^\downarrow \\
		2c-z_3 + z_2 & \leq \mu_3^\downarrow \\
		2b-z_2 + z_1 & \leq \mu_2^\downarrow \\
		2a - z_1 & \leq \mu_1^\downarrow \\
		z_1,z_2,z_3 & \geq 0 \\
		\begin{bmatrix}a & b \\ b & c\end{bmatrix} & \succeq O.
	\end{align*}
	Similarly eliminating $z_3$ via
	\[
		\max\big\{0, 2c + z_2 - \mu_3^\downarrow\big\} \leq z_3 \leq 2b + \sum_{j=4}^{2n}\mu_{j}^\downarrow
	\]
	gives
	\begin{align*}
		-2b & \leq \sum_{j=k}^{2n}\mu_{j}^\downarrow \quad \text{for} \quad k = 4,\ldots,2n \\
		2c-2b + z_2 & \leq \sum_{j=3}^{2n}\mu_{j}^\downarrow \\
		2b-z_2 + z_1 & \leq \mu_2^\downarrow \\
		2a - z_1 & \leq \mu_1^\downarrow \\
		z_1,z_2 & \geq 0 \\
		\begin{bmatrix}a & b \\ b & c\end{bmatrix} & \succeq O.
	\end{align*}
	Then eliminating $z_2$ via
	\[
		\max\big\{0, 2b + z_1 - \mu_2^\downarrow\big\} \leq z_2 \leq 2b - 2c + \sum_{j=3}^{2n}\mu_{j}^\downarrow
	\]
	gives
	\begin{align*}
		-2b & \leq \sum_{j=k}^{2n}\mu_{j}^\downarrow \quad \text{for} \quad k = 4,\ldots,2n \\
		2c-2b & \leq \sum_{j=3}^{2n}\mu_{j}^\downarrow \\
		2c + z_1 & \leq \sum_{j=2}^{2n}\mu_{j}^\downarrow \\
		2a - z_1 & \leq \mu_1^\downarrow \\
		z_1 & \geq 0 \\
		\begin{bmatrix}a & b \\ b & c\end{bmatrix} & \succeq O.
	\end{align*}
	Finally, eliminating $z_1$ via
	\[
		\max\big\{0, 2a - \mu_1^\downarrow\big\} \leq z_1 \leq -2c + \sum_{j=2}^{2n}\mu_{j}^\downarrow
	\]
	gives
	\begin{align}\begin{split}\label{sdp:2n_exist_constraints_final}
		-2b & \leq \sum_{j=k}^{2n}\mu_{j}^\downarrow \quad \text{for} \quad k = 4,\ldots,2n \\
		2c-2b & \leq \sum_{j=3}^{2n}\mu_{j}^\downarrow \\
		2c & \leq \sum_{j=2}^{2n}\mu_{j}^\downarrow \\
		2a+2c & \leq \sum_{j=1}^{2n}\mu_{j}^\downarrow \\
		\begin{bmatrix}a & b \\ b & c\end{bmatrix} & \succeq O.
	\end{split}\end{align}
	
	To simplify this system of inequality even further, we we recall that $s_k = \sum_{j=k}^{2n} \mu_j^\downarrow$, which are exactly the right-hand-sides of the inequalities. Also, we can add the constraint $-2b \leq s_3$ without affecting anything, since it is implied by the existing constraint $2c - 2b \leq s_3$. Then we can replace the set of constraints $-2b \leq s_k$ for $k = 3, 4, \ldots, 2n$ by the single constraint $-2b \leq s_{-}$, since $s_{-} = \min_k\{s_k\}$, and we know that this minimum is not attained when $k = 1$ or $k = 2$ since the third and fourth inequalities in~\eqref{sdp:2n_exist_constraints_final} tell us that $s_1 \geq s_2 \geq 0$. After making these simplifications, the system of inequalities has the form
	\begin{align}\begin{split}\label{sdp:2n_exist_constraints_finalb}
		-2b & \leq s_{-} \\
		2c-2b & \leq s_3 \\
		2c & \leq s_2 \\
		2a+2c & \leq s_1 \\
		\begin{bmatrix}a & b \\ b & c\end{bmatrix} & \succeq O.
	\end{split}\end{align}
	To finally get this system of inequalities into the form described by part~(c) of the theorem, we define $X = \begin{bmatrix}
		2a & -2b \\ -2b & 2c
	\end{bmatrix}$ and note that $X$ is positive semidefinite if and only if $\begin{bmatrix}
	a & b \\ b & c
	\end{bmatrix}$ is positive semidefinite. This completes the proof of the equivalence of~(b) and~(c).
	
	We now show that conditions (c) and (d) of the theorem are equivalent. The proof of this fact is similarly in flavor to the proof of equivalence of (b) and (c)---we just use quantifier elimination techniques to eliminate the variables $a$, $b$, and $c$ from the system of inequalities~\eqref{sdp:2n_exist_constraints_finalb}. First, note that the positive semidefinite constraint is equivalent to $a,c,ac-b^2 \geq 0$. However, note that if either $a = 0$ or $c = 0$ then $b = 0$, in which case the system is simply equivalent to the requirement that $\mu_j \geq 0$ for all $j$, which implies that the inequalities in condition~(d) hold. Thus we assume from now on that $a,c \neq 0$.
	
	Similarly, notice that if $b \leq 0$ then the first inequality in~\eqref{sdp:2n_exist_constraints_finalb} implies $\mu_{j} \geq 0$ for all $j$, which satisfies condition~(d) of the theorem. Thus we assume without loss of generality that $b > 0$. Well, if that system of inequalities is true for a particular $a > 0$ then certainly it is true if we increase $a$ even further (subject to the contraint $2a+2c \leq s_1$). Thus we may in fact assume that $a = s_1/2 - c$. After making these simplifications, the system of inequalities has the form
	\begin{align}\begin{split}\label{sdp:2n_exist_c_to_d_1}
	-2b & \leq s_{-} \\
	2c-2b & \leq s_3 \\
	2c & \leq s_2 \\
	b^2 & \leq c(s_1/2 - c).
	\end{split}\end{align}
	
	Again, we might as well increase $b$ up to $\sqrt{c(s_1/2-c)}$, so we assume this from now on. This gives us the equivalent system
	\begin{align}\begin{split}\label{sdp:2n_exist_c_to_d_2}
		-\sqrt{2c(s_1-2c)} & \leq s_{-} \\
		2c-\sqrt{2c(s_1-2c)} & \leq s_3 \\
		2c & \leq s_2.
	\end{split}\end{align}
	
	Well, the first inequality above is equivalent to
	\[
		2c(s_1-2c) \geq s_{-}^2,
	\]
	which is a quadratic equation in $c$ that is equivalent to the quadratic having two real roots, and $c$ being between those roots: 
	\begin{align*}
		-2s_{-} & \leq s_1 \\
		s_1 - \sqrt{s_1^2 - 4s_{-}^2} & \leq 4c \leq s_1 + \sqrt{s_1^2 - 4s_{-}^2}
	\end{align*}
	The second inequality in the system~\eqref{sdp:2n_exist_c_to_d_2} is a bit more inconvenient to deal with, since $s_3-2c$ might be positive or negative, so we have to be more careful when squaring the inequality. Specifically, we find that it is equivalent to
	\begin{align*}
		2c \leq s_3 \quad \text{or} \quad 2c(s_1-2c) & \geq (2c - s_3)^2.
	\end{align*}
	By expanding the quadratic and solving for $c$, we see that it is equivalent to
	\begin{align*}
		(s_1 + 2s_3)^2 \geq 8s_3^2 \quad \text{and} \quad (s_1 + 2s_3) - \sqrt{(s_1+2s_3)^2 - 8s_3^2} \leq 8c \leq (s_1 + 2s_3) + \sqrt{(s_1+2s_3)^2 - 8s_3^2}.
	\end{align*}
	Thus the system of inequalities~\eqref{sdp:2n_exist_c_to_d_2} is equivalent to \emph{at least one} of the following two systems of inequalities holding:
	\begin{align}\begin{split}\label{sdp:2n_exist_c_to_d_part1_1}
	-2s_{-} & \leq s_1 \\
	s_1 - \sqrt{s_1^2 - 4s_{-}^2} \leq 4c & \leq s_1 + \sqrt{s_1^2 - 4s_{-}^2} \\
	2c & \leq s_2 \\
	2c & \leq s_3 \\
	\end{split}\end{align}
	or
	\begin{align}\begin{split}\label{sdp:2n_exist_c_to_d_part2_1}
	-2s_{-} & \leq s_1 \\
	s_1 - \sqrt{s_1^2 - 4s_{-}^2} \leq 4c & \leq s_1 + \sqrt{s_1^2 - 4s_{-}^2} \\
	2c & \leq s_2 \\
	8s_3^2 & \leq (s_1 + 2s_3)^2 \\
	(s_1 + 2s_3) - \sqrt{(s_1+2s_3)^2 - 8s_3^2} \leq 8c & \leq (s_1 + 2s_3) + \sqrt{(s_1+2s_3)^2 - 8s_3^2}.
	\end{split}\end{align}
	
	Well, to eliminate $c$ from the system of inequalities~\eqref{sdp:2n_exist_c_to_d_part1_1}, we note that $2c \leq s_3$ implies $s_3 \geq 0$, so $\mu_3^\downarrow \geq 0$, so $s_2 \geq s_3$, which implies that we can discard the $2c \leq s_2$ inequality, as it is redundant. Then when we ``squeeze out'' $4c$, we are left with the equivalent system
	\begin{align}\begin{split}\label{sdp:2n_exist_c_to_d_part1_2}
		-2s_{-} & \leq s_1 \\
		s_1 - \sqrt{s_1^2 - 4s_{-}^2} & \leq 2s_3.
	\end{split}\end{align}
	
	On the other hand, if we ``squeeze out'' $c$ from the system of inequalities~\eqref{sdp:2n_exist_c_to_d_part2_1}, we are left with the much uglier system (where we recall that $q_1 := s_1^2 - 4s_{-}^2$ and $q_2 := (s_1+2s_3)^2 - 8s_3^2$):
	\begin{align}\begin{split}\label{sdp:2n_exist_c_to_d_part2_2}
		q_1,q_2 & \geq 0 \\
		\sqrt{q_1} & \geq s_1 - 2s_2 \\
		\sqrt{q_2} & \geq s_1 - 4s_2 + 2s_3 \\
		2\sqrt{q_1} + \sqrt{q_2} & \geq |s_1 - 2s_3|.
	\end{split}\end{align}
	
	Thus condition~(c) of the theorem is equivalent to \emph{at least one} of the systems of inequalities~\eqref{sdp:2n_exist_c_to_d_part1_2} and~\eqref{sdp:2n_exist_c_to_d_part2_2} holding. To remove the absolute value bars from system~\eqref{sdp:2n_exist_c_to_d_part2_2}, we first note that if $s_3 < 0$ then $2s_3 - s_1 < 0$, so the inequality $2\sqrt{q_1} + \sqrt{q_2} \geq 2s_3 - s_1$ is trivial. On the other hand, if $s_3 \geq 0$ then we observe that $s_1 \geq s_2 \geq 0$, which implies $\mu_2^\downarrow \geq 0$, so $s_1 \geq s_3$. Thus $q_2 = (s_1 + 2s_3)^2 - 8s_3^2 = s_1^2 + 4s_3(s_1 - s_3) \geq s_1^2$. It follows that
	\[
		2\sqrt{q_1} + \sqrt{q_2} \geq \sqrt{q_2} \geq s_1 \geq s_1 + 2(s_3 - s_1) = 2s_3 - s_1,
	\]
	as desired.
	
	To complete the proof, we need to show that the system of inequalities~\eqref{sdp:2n_exist_c_to_d_part1_2} holding implies that the system~\eqref{sdp:2n_exist_c_to_d_part2_2} holds, so we can completely discard the system~\eqref{sdp:2n_exist_c_to_d_part1_2}. Well, if $-2s_{-} \leq s_1$ then $q_1 \geq 0$. The second inequality in~\eqref{sdp:2n_exist_c_to_d_part1_2} says exactly that $\sqrt{q_1} \geq s_1 - 2s_3$, and since $s_2 \geq s_3$ this implies both $\sqrt{q_1} \geq s_1 - 2s_2$ and $2\sqrt{q_1} + \sqrt{q_2} \geq s_1 - 2s_3$. All that remains is to show that $q_2 \geq 0$ and $\sqrt{q_2} \geq s_1 - 4s_2 + 2s_3$.
	
	Well, to see that $q_2 \geq 0$, we note that $s_1 - \sqrt{s_1^2 - 4s_{-}^2} \leq 2s_3$ implies $s_3 \geq 0$, in which case we have $q_2 = s_1^2 + 4s_3(s_1-s_3) \geq s_1^2 \geq 0$. Thus we also conclude that $\sqrt{q_2} \geq s_1 \geq s_1 - 2s_2 \geq s_1 - 2s_2 - 2(s_2 - s_3) = s_1 - 4s_2 + 2s_3$, which completes the proof that the system~\eqref{sdp:2n_exist_c_to_d_part1_2} implies the system~\eqref{sdp:2n_exist_c_to_d_part2_2}. This finally gets us to the form described in part~(d) of the theorem, and completes the proof.
\end{proof}

Notice in particular that conditions~(c) and~(d) provide necessary conditions that the spectra of entanglement witnesses in $\textup{BP}_{2,n}$ must satisfy. However, because $\sigma(\textup{BP}_{2,n})$ is not convex (see Example~\ref{exam:eigs_not_convex}), these conditions are not sufficient. A comparison of these necessary conditions with the exact result of Theorem~\ref{thm:2x2_eigs} in the two-qubit case is provided by Figure~\ref{fig:two_qubit_region}.


\begin{figure}[htb]
	\centering
	\includegraphics[width=7cm]{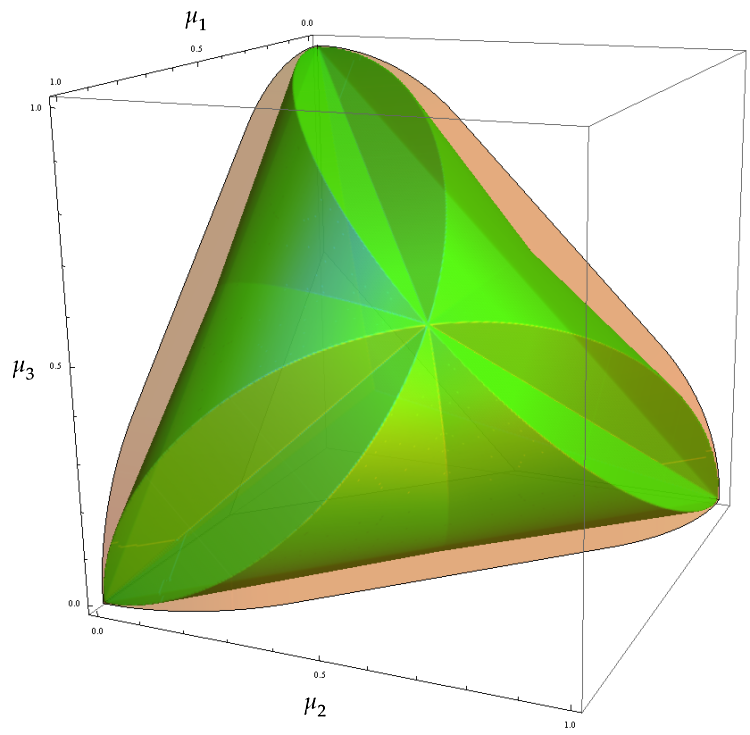}
	\caption{The region of tuples $(\mu_1,\mu_2,\mu_3) \in \mathbb{R}^3$ with the property that $(\mu_1,\mu_2,\mu_3,\mu_4)$ is the spectrum of an entanglement witness, where $\mu_4 = 1-\mu_1-\mu_2-\mu_3$, is displayed in green (as in Figure~\ref{fig:two_qubit_ew}). The orange region depicts the extra points that are not spectra of entanglement witnesses, but cannot be ruled out by the necessary conditions of Theorem~\ref{thm:mu2n_characterize}. The orange region looks like a convex ``shield'' in front of the non-convex green region.}\label{fig:two_qubit_region}
\end{figure}

\begin{example}\label{ex:23_dim_eg}
	Consider the vector of eigenvalues $\vec{\mu} = ((1+\sqrt{5}/2), (1+\sqrt{5}/2), 1, 1, c, c)$ as the spectrum of an entanglement witness in $M_2(\mathbb{C}) \otimes M_3(\mathbb{C})$, where $0 > c \in \mathbb{R}$. We can determine a bound on the most negative eigenvalue $c$ by plugging $\vec{\mu}$ into part~(d) of Theorem~\ref{thm:mu2n_characterize}. In this case, the tightest restrictions on $c$ come from the inequality $q_1 \geq 0$, so this is the only one we explicitly work through. First, we compute $q_1$ in terms of $c$:
	\begin{align*}
		q_1 &= s_1^2 - 4s_-^2 = 4(c^2 + (3+\sqrt{5})c +(14 +6 \sqrt{5})/4) - 16c^2 \\
		&= -12c^2 + 4(3+\sqrt{5})c +14 + 6\sqrt{5} \geq 0
	\end{align*}
	The above quadratic inequality limits $c$ to being between its two roots, which are
	\begin{align*}
		\frac{ -4(3+\sqrt{5}) \pm \sqrt{16(14+\sqrt{5}) -4(-12)(14+6\sqrt{5}) } }{2(-12)} = \frac{ (3+\sqrt{5}) \pm 2(3+\sqrt{5}) }{6}.
	\end{align*}
	It follows that $c$ cannot be smaller than the lesser of these two roots---i.e., $\vec{\mu} \in \mathrm{Conv}\big(\sigma(\textup{BP}_{2,3})\big)$ if and only if $c \geq -(3+\sqrt{5})/6$.
	
	Alternatively, we can use part~(c) of Theorem~\ref{thm:mu2n_characterize} to see that $\vec{\mu} \in \mathrm{Conv}\big(\sigma(\textup{BP}_{2,3})\big)$ when $c \geq -(3+\sqrt{5})/6$, since in this case we can choose the positive semidefinite matrix $X \in M_2(\mathbb{R})$ to be
	\[
		X = \begin{bmatrix}
			(1+\sqrt{5})/2 & c \\ c & (1+\sqrt{5})/2 + 2c + 2
		\end{bmatrix},
	\]
	and it is straightforward to verify this matrix satisfies the conditions of part~(c) of the theorem.
	
	It is worth noting that we explicitly constructed an entanglement witness with spectrum $\vec{\mu}$ when $c = (1-\sqrt{5})/2$ in Example~\ref{ex:23_2_neg}. However, we have now shown that the minimal value of $c$ is $-(3+\sqrt{5})/6 \approx -0.8727 < -0.6180 \approx (1-\sqrt{5})/2$. We have not been able to explicitly construct an entanglement witness with spectrum $\vec{\mu}$ when $c = -(3+\sqrt{5})/6$, and in fact one might not even exist since $\sigma(\textup{BP}_{2,3}) \neq \mathrm{Conv}\big(\sigma(\textup{BP}_{2,3})\big)$. However, we \emph{can} explicitly demonstrate that $\vec{\mu} \in \mathrm{Conv}\big(\sigma(\textup{BP}_{2,3})\big)$ in this case, since Lemma~\ref{lem:eigenvalues_rank1_pt} tells us that $(1,1,1,0,0,-1) \in \sigma(\textup{BP}_{2,3})$, so
	\[
		(1,1,1,0,0,-1) + (1,1,0,1,-1,0) = (2,2,1,1,-1,-1) \in \mathrm{Conv}\big(\sigma(\textup{BP}_{2,3})\big),
	\]
	which implies that if $c = -(3+\sqrt{5})/6$ then
	\begin{align*}
		\frac{\sqrt{5}}{3}(2,2,1,1,-1,-1) + \frac{3-\sqrt{5}}{6}(1,1,2,2,-1,-1) & \\
		= ((1+\sqrt{5})/2,(1+\sqrt{5})/2,1,1,c,c) & \in \mathrm{Conv}\big(\sigma(\textup{BP}_{2,3})\big).
	\end{align*}
\end{example}

\section{Higher-Dimensional Entanglement Witnesses}\label{sec:high_dim_ews}

We now consider the case of entanglement witnesses in arbitrary dimensions. This case is much more difficult for two reasons. First, we do not know whether or not $\textup{ASEP}_{m,n} = \textup{APPT}_{m,n}$ when $m,n \geq 3$ \cite{AJR15}, and thus we similarly do not know that $\mathrm{Conv}\big(\sigma(\textup{BP}_{m,n})\big) = \mathrm{Conv}\big(\sigma(\textup{DBP}_{m,n})\big)$. Second, the characterization of $\textup{APPT}_{m,n}$ is quite complicated and requires further explanation when $\min\{m,n\} \geq 3$. For simplicity, we assume for the remainder of this section that $m \leq n$.

We now provide more details about the characterization of absolutely PPT states, but for a full and rigorous description, we refer the interested reader to \cite{Hil07}. We start by constructing several linear maps $L_j : M_m(\mathbb{R}) \rightarrow \mathbb{R}^{mn}$. To illustrate how these linear maps $\{L_j\}$ are constructed, recall from Lemma~\ref{lem:eigenvalues_rank1_pt} that if a pure state $\ket{v} \in \mathbb{C}^m \otimes \mathbb{C}^n$ has Schmidt coefficients $\alpha_1 \geq \alpha_2 \geq \cdots \geq \alpha_m \geq 0$, then the eigenvalues of $(\ketbra{v}{v})^\Gamma$ are the numbers $\alpha_j^2$ ($1 \leq j \leq m$), $\pm\alpha_i\alpha_j$ ($1 \leq i \neq j \leq m$), and possibly some extra zeroes. Well, let's consider the possible orderings of those numbers. For example, if $m = 2$ then the only possible ordering is $\alpha_1^2 \geq \alpha_1\alpha_2 \geq \alpha_2^2 \geq 0 \geq \cdots \geq 0 \geq -\alpha_1\alpha_2$, whereas if $m = 3$ then there are two possible orderings:
\begin{align*}
	\alpha_1^2 & \geq \alpha_1\alpha_2 \geq \alpha_1\alpha_3 \geq \alpha_2^2 \geq \alpha_2\alpha_3 \geq \alpha_3^2 \geq 0 \geq \cdots \geq 0 \geq -\alpha_2\alpha_3 \geq -\alpha_1\alpha_3 \geq -\alpha_1\alpha_2 \quad \text{or} \\
	\alpha_1^2 & \geq \alpha_1\alpha_2 \geq \alpha_2^2 \geq \alpha_1\alpha_3 \geq \alpha_2\alpha_3 \geq \alpha_3^2 \geq 0 \geq \cdots \geq 0 \geq -\alpha_2\alpha_3 \geq -\alpha_1\alpha_3 \geq -\alpha_1\alpha_2.
\end{align*}

Well, for each of these orderings, we associate a linear map $L_j : M_m(\mathbb{R}) \rightarrow \mathbb{R}^{mn}$ by placing $\pm y_{i,j}$ into the position of $L_j(Y)$ where $\pm \alpha_i\alpha_j$ appears in the associated ordering (and actually $L_j$ is just a linear map on \emph{symmetric} matrices, not all matrices, so that we do not have to worry about distinguishing between $y_{i,j}$ and $y_{j,i}$). For example, in the $m = 2$ case, there is just one linear map $L_1 : M_2(\mathbb{R}) \rightarrow \mathbb{R}^{2n}$, and it is
\[
	L_1\left(\begin{bmatrix}y_{1,1} & y_{1,2} \\ y_{1,2} & y_{2,2}\end{bmatrix}\right) = (y_{1,1},y_{1,2},y_{2,2},0,\ldots,0,-y_{1,2}).
\]
Similarly, in the $m = 3$ case there are two linear maps $L_1,L_2 : M_3(\mathbb{R}) \rightarrow \mathbb{R}^{3n}$, and they are
\begin{align}\begin{split}\label{eq:3x3_L_maps}
	L_1\left(\begin{bmatrix}y_{1,1} & y_{1,2} & y_{1,3} \\ y_{1,2} & y_{2,2} & y_{2,3} \\ y_{1,3} & y_{2,3} & y_{3,3}\end{bmatrix}\right) & = (y_{1,1},y_{1,2},y_{1,3},y_{2,2},y_{2,3},y_{3,3},0,\ldots,0,-y_{2,3},-y_{1,3},-y_{1,2}) \quad \text{and} \\
	L_2\left(\begin{bmatrix}y_{1,1} & y_{1,2} & y_{1,3} \\ y_{1,2} & y_{2,2} & y_{2,3} \\ y_{1,3} & y_{2,3} & y_{3,3}\end{bmatrix}\right) & = (y_{1,1},y_{1,2},y_{2,2},y_{1,3},y_{2,3},y_{3,3},0,\ldots,0,-y_{2,3},-y_{1,3},-y_{1,2}).
\end{split}\end{align}

The number of distinct possible orderings (and thus the number of linear maps to be considered) grows exponentially in $m$. For $m = 2, 3, 4, \ldots$, this quantity equals $1, 2, 10, 114, 2608, 107498, \ldots$, though no formula is known for computing it in general \cite{oeisA237749}.

The main result of \cite{Hil07} says that $\vec{\lambda} \in \textup{APPT}_{m,n}$ if and only if $L_j^*(\vec{\lambda}^\uparrow) \succeq O$ for all $j$, where $\vec{\lambda}^\uparrow$ is the vector with the same entries as $\vec{\lambda}$ sorted in \emph{non-decreasing} order (rather than the usual non-increasing order we have used previously in this paper). For example, when $m = n = 3$, the result says that $\vec{\lambda} \in \textup{APPT}_{3,3}$ if and only if $L_1^*(\vec{\lambda}^\uparrow) \succeq O$ and $L_2^*(\vec{\lambda}^\uparrow) \succeq O$, where
\begin{align*}
L_1^*(\vec{\lambda}^\uparrow) = \begin{bmatrix}
2\lambda_9^\downarrow & \lambda_8^\downarrow - \lambda_1^\downarrow & \lambda_7^\downarrow - \lambda_2^\downarrow \\
\lambda_8^\downarrow - \lambda_1^\downarrow & 2\lambda_6^\downarrow & \lambda_5^\downarrow - \lambda_3^\downarrow \\
\lambda_7^\downarrow - \lambda_2^\downarrow & \lambda_5^\downarrow - \lambda_3^\downarrow & 2\lambda_4^\downarrow
\end{bmatrix} \quad \text{and} \quad L_2^*(\vec{\lambda}^\uparrow) = \begin{bmatrix}
2\lambda_9^\downarrow & \lambda_8^\downarrow - \lambda_1^\downarrow & \lambda_6^\downarrow - \lambda_2^\downarrow \\
\lambda_8^\downarrow - \lambda_1^\downarrow & 2\lambda_7^\downarrow & \lambda_5^\downarrow - \lambda_3^\downarrow \\
\lambda_6^\downarrow - \lambda_2^\downarrow & \lambda_5^\downarrow - \lambda_3^\downarrow & 2\lambda_4^\downarrow
\end{bmatrix}.
\end{align*}

The last tool that we need to be able to state our characterization of $\mathrm{Conv}\big(\sigma(\textup{DBP}_{m,n})\big)$ is a function $p : \mathbb{R}^n \rightarrow \mathbb{R}^n$ that computes the partial sums of a vector:
\[
	p(\vec{v}) \defeq \left(\sum_{j=1}^n v_j, \sum_{j=2}^n v_j, \sum_{j=3}^n v_j, \ldots, \sum_{j=n-1}^n v_j, v_n\right).
\]
\begin{theorem}\label{thm:mumn_characterize}
	Suppose $\vec{\mu} \in \mathbb{R}^{mn}$ and $m \leq n$. Then the following are equivalent:
	\begin{enumerate}
		\item[a)] $\vec{\mu} \in \mathrm{Conv}\big(\sigma(\textup{DBP}_{m,n})\big)$.
		
		\item[b)] There exist positive semidefinite matrices $Y_1,Y_2,\ldots \in M_m(\mathbb{R})$ such that
		\begin{align*}
			\sum_j p\big(L_j(Y_j)\big) \leq p\big(\mu^\downarrow\big).
		\end{align*}
	\end{enumerate}
\end{theorem}

Before proving this theorem, we note that in the $m = 2$ case, it reduces exactly to the equivalence of conditions~(b) and~(c) in Theorem~\ref{thm:mu2n_characterize}. Also, the in previous sections we had defined a vector $\vec{s}$ containing that partial sums of $\mu^\downarrow$---this vector $\vec{s}$ is exactly $p(\mu^\downarrow)$.

\begin{proof}
	Again, we start by using Corollary~\ref{cor:conv_hull_duality_appt} to note that $\vec{\mu} \in \mathrm{Conv}\big(\sigma(\textup{DBP}_{m,n})\big)$ if and only if $\vec{\mu} \in \textup{APPT}_{m,n}^\circ$, if and only if the optimal value of the following semidefinite program is non-negative:
	\begin{align}\begin{split}\label{sdp:mn_primal}
	\text{minimize: } & \ \sum_{j=1}^n \mu_j^\downarrow \lambda_{n+1-j} \\
	\text{subject to: } & \ L_j^*(\vec{\lambda}^\uparrow) \succeq O \quad \text{for all $j$} \\
	& \ \sum_{j=1}^{mn} \lambda_j = 1 \\
	& \ \lambda_1 \geq \lambda_2 \geq \cdots \geq \lambda_{mn} \geq 0.
	\end{split}\end{align}
	
	The dual of this semidefinite program is as follows:
	\begin{align}\begin{split}\label{sdp:mn_dual}
		\text{maximize: } & \ d \\
		\text{subject to: } & \ \sum_j L_j(Y_j) - (\vec{z},0) + (0,\vec{z}) + d\vec{1} \leq \vec{\mu}^\downarrow \\
		& \ 0 \leq d \in \mathbb{R}, \quad \vec{0} \leq \vec{z} \in \mathbb{R}^{mn-1} \\
		& \ O \preceq Y_j \in M_m(\mathbb{R}) \quad \text{for all $j$},
	\end{split}\end{align}
	where $\vec{1} \in \mathbb{R}^{mn}$ is the vector with each entry equal to $1$. Well, strong duality holds for this semidefinite program for reasons almost identical to those provided in the proof of Theorem~\ref{thm:mu2n_characterize}, so we just want to determine whether or not this dual problem has a feasible point with $d \geq 0$. In other words, we want to know whether or not there exists $\vec{0} \leq \vec{z} \in \mathbb{R}^{mn-1}$ and positive semidefinite matrices $\{Y_j\} \subseteq M_m(\mathbb{R})$ such that
	\begin{align}\label{sdp:mn_dual_exists}
		\sum_j L_j(Y_j) - (\vec{z},0) + (0,\vec{z}) \leq \vec{\mu}^\downarrow.
	\end{align}
	Well, applying the same quantifier elimination techniques from the proof of Theorem~\ref{thm:mu2n_characterize} to the vector $\vec{z}$ shows that inequality~\eqref{sdp:mn_dual_exists} is equivalent to the existence of positive semidefinite matrices $\{Y_j\} \subseteq M_m(\mathbb{R})$ such that
	\begin{align}\label{sdp:mn_dual_exists_done}
		\sum_j p(L_j(Y_j)) \leq p(\vec{\mu}^\downarrow),
	\end{align}
	which completes the proof.
\end{proof}	

For example, in the $m = n = 3$ case, we recall the two maps $L_1,L_2$ from Equation~\eqref{eq:3x3_L_maps}. After working through the details, we see that Theorem~\ref{thm:mumn_characterize} says that $\vec{\mu} \in \mathrm{Conv}\big(\sigma(\textup{DBP}_{3,3})\big)$ if and only if there exist $3 \times 3$ positive semidefinite matrices $X,Y \in M_3(\mathbb{R})$ such that
\begin{align}\begin{split}\label{ineq:sys_33}
	(x_{1,1} + x_{2,2}+x_{3,3}) + (y_{1,1} + y_{2,2} + y_{3,3}) & \leq \mu_1^\downarrow + \mu_2^\downarrow +\cdots + \mu_9^\downarrow \\
	(x_{2,2}+x_{3,3}) + (y_{2,2} + y_{3,3}) & \leq \mu_2^\downarrow + \mu_3^\downarrow +\cdots + \mu_9^\downarrow \\
	(x_{2,2}+x_{3,3} - x_{1,2}) + (y_{2,2} + y_{3,3}-y_{1,2}) & \leq \mu_3^\downarrow + \mu_4^\downarrow +\cdots + \mu_9^\downarrow \\
	(x_{3,3} - x_{1,2}) + (y_{2,2} + y_{3,3}-y_{1,2} - y_{1,3}) & \leq \mu_4^\downarrow + \mu_5^\downarrow +\cdots + \mu_9^\downarrow \\
	(x_{3,3} - x_{1,2}-x_{1,3}) + (y_{3,3}-y_{1,2} - y_{1,3}) & \leq \mu_5^\downarrow + \mu_6^\downarrow +\cdots + \mu_9^\downarrow \\
	(x_{3,3} - x_{1,2}-x_{1,3}-x_{2,3}) + (y_{3,3}-y_{1,2} - y_{1,3} - y_{2,3}) & \leq \mu_6^\downarrow + \mu_7^\downarrow + \mu_8^\downarrow + \mu_9^\downarrow \\
	(-x_{1,2}-x_{1,3}-x_{2,3}) + (-y_{1,2} - y_{1,3} - y_{2,3}) & \leq \mu_7^\downarrow + \mu_8^\downarrow +\mu_9^\downarrow \\
	(-x_{1,2}-x_{1,3}) + (-y_{1,2} - y_{1,3}) & \leq \mu_8^\downarrow + \mu_9^\downarrow \\
	-x_{1,2} - y_{1,2} & \leq \mu_9^\downarrow \\
\end{split}\end{align}

Note in particular that the inequality involving $\mu_4^\downarrow + \cdots + \mu_9^\downarrow$ is not symmetric in $X$ and $Y$, and this inequality is the reason that we require \emph{two} positive semidefinite matrices $X$ and $Y$ rather than just one. In general, the number of inequalities to be checked is $\min\{m,n\}^2$, but the number of positive semidefinite matrices invovled (and the number of terms being added up in each inequality) grows exponentially in $\min\{m,n\}$ (as we discussed earlier).

While this system of inequalities is not the type of thing that we expect to be able to solve analytically like how we did in the qubit-qudit cases (condition~(d) of Theorem~\ref{thm:mu2n_characterize}), the existence of $X$ and $Y$ can be determined numerically via semidefinite programming (e.g., in the CVX package for MATLAB \cite{CVX}). Thus these inequalities can be used computationally as necessary conditions that eigenvalues of entanglement witnesses must satisfy.

\begin{example}\label{exam:3x3_case}
	Suppose $c \in \mathbb{R}$ and consider the vector of eigenvalues $\vec{\mu} = (1,1,1,1,1,1,-1,-1,c)$. To determine which values of $c$ result in $\vec{\mu} \in \mathrm{Conv}\big(\sigma(\textup{DBP}_{3,3})\big)$, we note that it is straightforward to verify that $\vec{\lambda} := (2,2,2,1,1,1,1,1,1)/12 \in \textup{APPT}_{3,3}$ and
	\begin{align*}
		\sum_{j=1}^9 \mu_j^\downarrow \lambda_{10-j} = (2c - 2 - 2 + 1 + 1 + 1 + 1 + 1 + 1)/12 = (c + 1)/6,
	\end{align*}
	which is negative when $c < -1$. We conclude from Corollary~\ref{cor:conv_hull_duality_appt} that $\vec{\mu} \not\in \textup{APPT}_{3,3}^\circ = \mathrm{Conv}\big(\sigma(\textup{DBP}_{3,3})\big)$ whenever $c < -1$ (and in particular, there does not exist a decomposable entanglement witness with eigenvalues $(1,1,1,1,1,1,-1,-1,c)$ when $c < -1$).
	
	On the other hand, we can see that $\vec{\mu} \in \mathrm{Conv}\big(\sigma(\textup{DBP}_{3,3})\big)$ whenever $c \geq -1$ via the system of inequalities~\eqref{ineq:sys_33}. In particular,
	\begin{align*}
		X = Y = \frac{1}{2}\begin{bmatrix}
			1 & 1 & 1 \\ 1 & 1 & 1 \\ 1 & 1 & 1
		\end{bmatrix}
	\end{align*}
	is positive semidefinite and satisfies the system of inequalities~\eqref{ineq:sys_33} when $c \geq -1$. By putting these two facts together, we see that $\vec{\mu} \in \mathrm{Conv}\big(\sigma(\textup{DBP}_{3,3})\big)$ if and only if $c \geq -1$.
	
	It is worth noting that we can also see directly that $\vec{\mu} \in \sigma(\textup{DBP}_{3,3})$ whenever $c \geq -1$, since if $\ket{v} = \ket{0}\otimes\ket{0} + \ket{1}\otimes\ket{1} + \ket{2}\otimes\ket{2}$ then $(\ketbra{v}{v})^\Gamma$ has eigenvalues $(1,1,1,1,1,1,-1,-1,-1)$, and we can increase the smallest eigenvalue by adding a suitable positive semidefinite matrix to $(\ketbra{v}{v})^\Gamma$.
\end{example}

Theorem~\ref{thm:mumn_characterize} only applies to decomposable entanglement witnesses, since we do not know whether or not $\textup{ASEP}_{m,n} = \textup{APPT}_{m,n}$ when $m,n \geq 3$. This result gives us a new avenue of approaching the absolute separability problem, since if we can find an entanglement witness with eigenvalues $\vec{\mu}$ such that $\vec{\mu} \not\in \mathrm{Conv}\big(\sigma(\textup{DBP}_{m,n})\big)$ then it follows that $\textup{ASEP}_{m,n} \subsetneq \textup{APPT}_{m,n}$. For example, in light of Example~\ref{exam:3x3_case} we know that if there exists an entanglement witness with eigenvalues $(1,1,1,1,1,1,-1,-1,c)$ for some $c < -1$, then $\textup{ASEP}_{3,3} \neq \textup{APPT}_{3,3}$.

\section{Conclusions}\label{sec:conclude}

We have introduced the inverse eigenvalue problem for the set of entanglement witnesses (or equivalently, the set of block-positive matrices). We completely solved this problem in the smallest-dimensional case of $M_2(\mathbb{C}) \otimes M_2(\mathbb{C})$ in Theorem~\ref{thm:2x2_eigs}, and we provided a strong set of necessary conditions for the $M_2(\mathbb{C}) \otimes M_n(\mathbb{C})$ in Theorem~\ref{thm:mu2n_characterize}(d). We then provided a general method for constructing necessary conditions for the spectra of \emph{decomposable} entanglement witnesses in arbitrary dimensions in Theorem~\ref{thm:mumn_characterize}, and we illustrated a duality relationship with absolute separability that provides a new line of attack for approaching the question of whether or not $\textup{ASEP}_{m,n} = \textup{APPT}_{m,n}$.

The most pressing open question resulting from this work is whether or not there exists an entanglement witness with spectrum $\vec{\mu} \not\in \mathrm{Conv}\big(\sigma(\textup{DBP}_{m,n})\big)$---finding such an entanglement witness would show that $\textup{ASEP}_{m,n} \subsetneq \textup{APPT}_{m,n}$. Another question that might be reasonably solvable is that of finding exact conditions for membership in $\sigma(\textup{BP}_{2,3})$. We know that every $W \in \textup{BP}_{2,3}$ can be written in the form $W = X^\Gamma + Y$, but the fact that $W$ can have two negative eigenvalues (and thus $X$ cannot be chosen to have rank $1$) makes it difficult to pin down exactly how negative its eigenvalues can be. For example, what are the restrictions on $c$ in Example~\ref{ex:23_dim_eg}? We showed in that example that $\vec{\mu} \in \mathrm{Conv}\big(\sigma(\textup{BP}_{2,3})\big)$ if and only if $c \geq -(3+\sqrt{5})/6 \approx -0.8727$. However, we have not been able to find an entanglement witness with spectrum $\vec{\mu}$ for any $c < (1-\sqrt{5})/2 \approx -0.6180$.\bigskip

\noindent\textbf{Acknowledgements.} N.J. was supported by {NSERC} Discovery Grant number RGPIN-2016-04003.

\bibliographystyle{ieeetr}
\bibliography{bib}

\end{document}